\documentclass[conference,a4paper]{IEEEtran}

\usepackage[utf8]{inputenc} 
\usepackage[T1]{fontenc}
\usepackage{url}              
\usepackage{cite}             

\usepackage[cmex10]{amsmath}  
\usepackage{mleftright}       
\mleftright                   

\usepackage{graphicx}         
\usepackage{booktabs}
\IEEEoverridecommandlockouts
\usepackage{cite}
\usepackage{amsmath,amssymb,amsfonts}
\usepackage{graphicx}
\usepackage{enumerate}
\usepackage{enumitem}
\usepackage{textcomp}
\usepackage{xcolor}
\usepackage{lipsum}
\usepackage{mathtools}
\usepackage{amsthm}
\usepackage{amsmath}
\usepackage{amssymb}
\usepackage{delarray}
\usepackage{bm}
\usepackage{graphicx}
\usepackage{color}
\usepackage{enumitem}
\usepackage{footmisc}
\usepackage{subfig}
\usepackage{graphicx}
\usepackage{cancel}
\usepackage[export]{adjustbox}
\usepackage{comment}
\usepackage{graphicx}
\usepackage{subfig}
\usepackage{caption}
\usepackage{algpseudocode}
\usepackage{algorithm2e}
\usepackage{pgfplots}
\usepackage{todonotes}
\pgfplotsset{compat=newest}
\pgfplotsset{plot coordinates/math parser=false}
\newlength\figureheight
\newlength\figurewidth

\newtheorem{theorem}{Theorem}

\newtheorem{lemma}{Lemma}
\newtheorem{prop}{Proposition}

\newtheorem{rem}{Remark}

\newcommand{\HHx}[4]{E\left(\N,\ell,k,\qq \right)}

\newcommand{\N}{n}

\makeatletter
\newcommand*{\rom}[1]{\expandafter\@slowromancap\romannumeral #1@}
\makeatother

\newcommand{\cU}{\mathcal{U}}

\newcommand{\LB}{L}

\newcommand{\Ls}{\mathcal{L}}
\newcommand{\Lst}[1]{\mathcal{L}_{#1}}

\newcommand{\qq}[0]{\alpha}

\newcommand{\nr}{s}

\newcommand{\mat}[1]{{\mathbf{B}}}
\newcommand{\vect}[1]{\mathbf{v}_{#1}}

\newcommand{\al}{\alpha}
\newcommand{\cM}{\mathsf{M}^{(c)}}
\newcommand{\pM}{\mathsf{M}^{(p)}}
\newcommand{\cy}{\bm{y}^{(c)}}
\newcommand{\sM}{\mathsf{M}^{(s)}}
\newcommand{\cq}{q}

\newcommand{\tht}{\theta}

\newcommand{\bx}{\bm{x}}
\newcommand{\supp}{\mathsf{supp}}

\renewcommand{\P}{\mathbb{P}}

\newcommand{\E}{\mathbb{E}}

\newcommand{\hh}{h}

\usepackage{centernot}

\newcommand{\Pp}{\mathsf{P}_+}
\newcommand{\Pn}{\mathsf{P}_-}
\newcommand{\Pe}{\mathsf{P}_\text{e}}
\newcommand{\cL}{\mathcal{L}}

\newcommand{\ytz}{\bm{y}_{\mathcal{T}_z}^{\ }}

\renewcommand{\comment}[1]{ }
\newcommand{\by}{\mathbf{y}}

\newcommand{\mf}{\mu_f}
\newcommand{\mm}{\mu_m}

\newcommand{\sQ}{\mathsf{Q}}
\newcommand{\sG}{\mathsf{G}}

\begin{document}

%
\title{Probabilistic Group Testing  in Distributed Computing with Attacked Workers}
\author{
\IEEEauthorblockN{Sarthak Jain, Martina Cardone,  Soheil Mohajer}
University of Minnesota, Minneapolis, \!MN 55455, \!USA,
\!Email: \{jain0122, mcardone, soheil\}@umn.edu
\\
\vspace{-0.9em}
\thanks{This research was supported in part by the U.S. National Science Foundation under Grants CCF-1907785 and CCF-2045237.} 
}
\maketitle

\begin{abstract}
 The problem of distributed matrix-vector product is considered, where the server distributes the task of the computation among $\N$ worker nodes, out of which $\LB$ are {compromised}
 {(but non-colluding)} and may return incorrect results. 
 {Specifically,} it is assumed that the 
 {compromised}
 workers are unreliable, that is, at any given time, each 
 {compromised}
 worker may return an incorrect and correct result with probabilities $\al$ and $1-\al$, respectively. Thus, the tests are noisy. 
 This work proposes a new probabilistic group testing approach to identify the 
 {unreliable/compromised}
 workers with $O\left(\frac{\LB \log(\N)}{\al}\right)$ tests. Moreover, using the proposed group testing method, sparse parity-check codes are constructed and used in the considered distributed computing framework for encoding, decoding and identifying the unreliable workers. This methodology has two distinct features: (i) 
 the cost of identifying the set of $\LB$ 
 unreliable 
 workers at the server can be shown to be considerably lower than existing distributed computing methods, 
 and (ii) the encoding and decoding functions are easily implementable and computationally efficient. 
\end{abstract}

\section{Introduction}
In distributed computing, an expensive task can be \emph{encoded} into sub-tasks of lower complexity, which can then be distributed to several worker nodes, in order to improve the overall computational speed~\cite{solanki19, tang2022, Yu2019, yu2020, Yu2017, yu_2017, jain2022}.
However, a subset of the workers may be {compromised}
and may return incorrect results of the sub-tasks assigned to them. {In other words,} these 
{compromised
workers 
are 
unreliable, i.e., they 
sometimes} return incorrect results {(i.e., they behave maliciously)} and at other times behave like reliable workers and return correct results, making them difficult to identify. In such scenarios, it is desirable that the encoding into sub-tasks is performed in a way that two properties hold: (i) \emph{Identification}: there is an efficient mechanism to identify the set of 
{unreliable/compromised}
worker nodes; and (ii) \emph{Decodability}: from the results of the sub-tasks assigned to the reliable worker nodes, the result of the original task can be recovered efficiently. 

In this work, we propose a probabilistic group testing based distributed computing scheme for the task of matrix-vector computation, which ensures both efficient identification and decodability. 
Group testing was first introduced in~\cite{dorfman1943} and extensively studied in areas ranging from medicine~\cite{verdun2021} to computer science~\cite{solanki19}, to efficiently identify a set of defective items in a large set of $\N$ items, by testing groups of items at a time rather than testing the items individually. Two types of group testing methodologies are prevalent: (i) \textit{probabilistic group testing}; and (ii) \textit{combinatorial group testing}. Combinatorial group testing is deterministic and identifies defective items with a zero probability of error\cite{Du1993}. In probabilistic group testing, this probability of error goes to $0$ as $\N \to \infty$, whereas for finite $\N$, it can be made arbitrarily small
by increasing the number of 
tests ~\cite{Chan2011,mazumdar2016,barg2017,inan2018}. 
Unlike traditional group testing, in this work, the goal is to identify {the set of $\LB$ (out of $\N$) unreliable workers, 
where the} unreliable workers do not always behave maliciously, but 
they sometimes hide their true identity.
A test can therefore be negative even if there were one or more unreliable workers in the tested group.  
Group testing with unreliable items {has been widely studied~\cite{atia2012boolean,cai2013grotesque, CheraghchiIT2011,  mazumdar2014group,jain2023}.} 
However, the probabilistic model for the behavior of unreliable workers 
that we {adopt 
is} different from the models considered in the above works. 
{This difference is primarily due to the fact that the behavior of an unreliable worker (i.e., either acting maliciously or acting reliably) remains the same for all the tests related to a specific computing task.
Thus, naively adopting the probabilistic} group testing methods of~\cite{atia2012boolean,cai2013grotesque, CheraghchiIT2011,  mazumdar2014group} for our model, 
produce sub-optimal results.

The contribution of this work is threefold: (i) we propose a probabilistic group testing scheme 
{to efficiently identify the $\LB$ unreliable workers (Section~\ref{section:system_model});} 
(ii) using the proposed group testing scheme, we construct sparse parity-check codes which are used for encoding {and enable} an efficient identification of the unreliable workers {(Section~\ref{section:group_testing})}; and (iii) we construct a low-complexity decoding function for retrieving the original matrix-vector product from the results of the reliable workers {(Section~\ref{section:group_testing})}. The proposed scheme can be shown to outperform 
the MDS-code based schemes presented in~\cite{solanki19,tang2022} for unreliable worker's identification.


{\noindent {\bf{Notation}.} For an integer $\N \geq 1$, we define $[\N] \triangleq \{1,2, \ldots,n\}$. For a matrix $\mathsf{M}$, we use $\mathsf{M}_{i,:}$ and $\mathsf{M}_{:,j}$ to represent its $i$th row and $j$th column, respectively. Moreover, for sets $\mathcal{A}$ and $\mathcal{B}$, $\mathsf{M}_{\mathcal{A},\mathcal{B}}$ {is the submatrix of $\mathsf{M}$ where only the rows in $\mathcal{A}$ and the columns in $\mathcal{B}$ are retained.} For $x,y \in \{0,1\}$, we define $x \succ y$ if $(x,y)=(1,0)$ and $x \preceq y$, otherwise. For a vector $\bm{z}$, we define $\supp(\bm{z}) = \{j:\bm{z}_j \neq 0\}$.}

\section{System Model}\label{section:system_model}
The server node aims at computing $T$ matrix-vector products $\mat{t} \cdot \vect{t}$ for $t \in [T]$, where $\mat{t} \in \mathbb{F}^{r \times c}$ and $\vect{t} \in \mathbb{F}^{c \times 1}$. In other words, we aim at computing the $t$th matrix-vector product $\mat{t}\cdot \vect{t}$ in time-slot $t\in [T]$. To speed up the computation, the server can distribute the task among $\N$ workers, denoted by the set $[\N]$. However, a set $\Ls \subseteq [\N]$ of $|\Ls| = \LB$ workers are unreliable and may return incorrect/noisy results for the task assigned to them. We model the random behavior of these $\LB$ unreliable workers as follows: at time-slot $t\in [T]$, each unreliable worker is \emph{attacked} with probability 
$\alpha$, and returns an incorrect result, independent of  other workers and other time-slots. If an unreliable worker remains unattacked in a time-slot $t$, it behaves reliably and returns the correct result. 
We denote by $\Lst{t} \subseteq \Ls$ the subset of unreliable workers, which are attacked  in time-slot $t$.
For any $w \in \Ls$, we have 
\begin{align} \label{eq:attacked}
    \mathbb{P}\left(w \in \Lst{t} \!\bigm|\! w \in \Ls\right) = \al.
\end{align}
For the distributed computing scheme, the vectors $\vect{t}$'s are first communicated to all the $\N$ workers. Moreover, each worker $w \in [\N]$ receives a matrix $\mathbf{W}^{(w)} \in \mathbb{F}^{\nr \times c}$ with $\nr=r/k\ll r$ rows (assuming that $k\in \mathbb{N}$ divides $r$) obtained by using a set of {encoding} functions ${f^{(w)}}: \mathbb{F}^{r \times c} \xrightarrow{} \mathbb{F}^{\nr \times c}$ as 
\begin{equation}
\label{eq:Encoding}
    \mathbf{W}^{(w)}={f^{(w)}}(\mat{t}), \qquad\qquad w \in [n].
\end{equation}
{For each $t \in [T]$,} the $w$th worker then computes the matrix-vector product $\mathbf{a}_t^{(w)} = {\mathbf{W}^{(w)}} \cdot \vect{t}$, and sends 
a message  
back to the server. In particular, the result $\tilde{\mathbf{a}}_t^{(w)}$ returned by the $w$th worker is given by
\begin{equation} \label{eq:z-channel}
\tilde{\mathbf{a}}_t^{(w)} = \begin{cases}
   {\mathbf{W}^{(w)}} \cdot \vect{t}  &  \text{if } w \notin \Lst{t}, \\
    {\mathbf{W}^{(w)}} \cdot \vect{t}+\mathbf{z}_t^{(w)} \qquad & \text{if } w \in \Lst{t},
    \end{cases}    
\end{equation}
where $\mathbf{z}_t^{(w)} \in \mathbb{F}^{s \times 1}$ is the noise vector corresponding to worker $w$ and is independent of the noise vectors of other attacked workers {(since we assume the unreliable workers to be non-colluding).} 
The 
{server}
is then required to identify the unreliable workers and subsequently reconstruct/decode the correct product $\mat{t} \cdot \vect{t}$ by applying a decoding function $g(\cdot)$ on the results received from (a subset of) the reliable workers. 

We propose a distributed computing scheme where the encoding functions $f^{(i)}$ in~\eqref{eq:Encoding} and the decoding function $g$ are such that the overall computational complexity at the server is low because: (i) the set of unreliable workers can be identified with a low computational cost of $O\left(\frac{\LB \log(\N)}{\al}\right) \times O\left(\frac{r\N}{\LB k}\right)$; and (ii) the decoding can be performed with a cost of $O\left(\frac{r \N}{k}\right) \times T$ operations for $T$ matrix-vector products.

\section{Group Testing for Identifying Unreliable Workers} \label{section:group_testing}
We are interested in identifying the set of all unreliable workers $\Ls$. 
Here, we {face} two main challenges: (1) workers cannot be tested individually, and (2) an unreliable worker is not always attacked and hence, we cannot identify it from its result in a single time-slot. 

Instead, we can test a group of workers, and 
a test result will be negative if all the unreliable workers in the selected group are unattacked during the time-slot of interest.  In other words, a test performed based on the results of time-slot $t$ will be positive if and only if there is at least one attacked unreliable worker $w \in \Lst{t}$ in the tested group. 

We let $m_t$ be the number of tests that are performed in time-slot $t\in[T]$. Then the total number of tests is ${M := \sum_{t=1}^T m_t}$. These tests can be represented using a \emph{contact} matrix ${\cM \in \{0,1\}^{M \times \N}}$, where $\cM_{j,w}=1$ if and only if the $w$th worker is included in the $j$th test. The tests conducted in time-slot $t$ correspond to the rows ${\mathcal{T}_t \!=\! [1+\sum_{j=1}^{t-1} m_t\!:\! \sum_{j=1}^{t} m_t]}$ of $\cM$. 
We can use a vector $\bx \in \{0,1\}^{n \times 1}$ to indicate whether or not each worker is  unreliable, that is, $\bx_w = 1$ if and only if $w \in \Ls$. 
Ideally, if all the unreliable workers were always attacked, i.e., $\alpha=1$, the result of the tests can be represented by a vector  $\cy \in \{0,1\}^{M}$ as
\begin{equation} \label{eq:define_y}
    \cy = \cM \odot \bm{x},
\end{equation}
where the multiplication and addition are logical \texttt{and} and \texttt{or}, respectively. More precisely, we have $\cy_i = \bigvee_{j=1}^n (\cM_{i,j} \land \bm{x}_j)$.

To capture the unreliable behavior of the {compromised workers}, we adopt the notation in~\cite{CheraghchiIT2011}, and define a \emph{sampling} matrix $\sM$ which is obtained from the contact matrix $\cM$~as 
\begin{equation} \label{eq:sampling0}
\sM_{\mathcal{T}_t,w} = \begin{cases}
   \cM_{\mathcal{T}_t,w}  &  \text{if } w \in \left([\N]\setminus \Ls\right) \cup \Lst{t}, \\
   \mathbf{0} & \text{if } w \in \Ls \setminus \Lst{t}.
    \end{cases}    
\end{equation}
Moreover, the result of the actual tests in the presence of unreliable workers is given by 
\begin{align}
    \bm{y} = \sM \odot \bm{x}.
\end{align}
{To understand how the sampling matrix models the unreliable behavior, let us assume that worker $w$ was selected to be included in a test $j \in \mathcal{T}_t$ in time-slot $t$. Then, we have $\cM_{j,w}=1$. If $w\in\Lst{t}$, then in time-slot $t$ the unreliable worker is attacked and exposes its true identity and therefore, we should have $\bm{y}^{(c)}_j = \bm{y}_j = 1$, which is achieved by setting $\sM_{j,w} = \cM_{j,w}$ in such a scenario. Now, assume that worker $w$ is unreliable (i.e., ${\bm{x}_w}=1$), but unattacked in time-slot $t$, that is, $w\in \Ls \setminus \Lst{t}$. In this case, $w$ hides its true identity. Since $\cM_{j,w}=1$, we have {$\cy_{j}=1$}. But the actual test result $\bm{y}_j$ should not be influenced by $w$. This unattacked behavior of worker $w$ in test $j \in \mathcal{T}_t$ is captured by setting $\sM_{j,w}=0$. }

Note that if $\alpha=1$, then $\Lst{t} = \Ls$ for all $t \in [T]$, and therefore $\cM_{j,w} = \sM_{j,w}$ for all $j \in [M]$ {and $w \in [n]$,} which reduces the problem to the classical group testing problem~\cite{dorfman1943}.

Examples of a contact matrix $\cM$ and a sampling matrix $\sM$ for $\N = 5$ workers, $T=2$, $m_1=2$, and $m_2 = 1$ is given in~\eqref{ex1}. Here, 
$\Ls = \{3,4\}$, and $\cL_1=\{4\}$.  Since worker~$3$ is unattacked at $t=1$, the entries of $\cM$ at column $3$ and $\mathcal{T}_1=[1:2]$ are replaced by $0$ in $\sM$: 
\begin{equation} \label{ex1}
    \cM \!=\! \begin{bmatrix}
1 & 0 & 0 & 1 & 0\\
0 & 1 & 1 & 0 & 0 \\
1 & 1 & 0 & 1 & 0
\end{bmatrix}, \ 
 \sM \!=\! \begin{bmatrix}
1 & 0 & 0 & 1 & 0\\
0 & 1 & 0 & 0 & 0 \\
1 & 1 & 0 & 1 & 0
\end{bmatrix}.
\end{equation}
\begin{rem}
{Note that} the tests are always conducted according to the contact matrix $\cM$, but due to the random  behavior of the unreliable workers, the actual result of the tests is determined by $\bm{y} = \sM \odot \bm{x}$ (rather than ${\cy = \cM \odot \bm{x}}$). 
\end{rem}

We seek to minimize the number of tests $M$, i.e., design a contact matrix $\cM$ with as few rows as possible, such that, given $\cM$ and $\bm{y}$, we can identify the unreliable set $\Ls$ (or equivalently the vector $\bm{x}$), with an arbitrarily small probability of error (as $\N\rightarrow \infty$).
In particular, 
we are interested in characterizing how the number of tests $M$ should scale with respect to the underlying parameters, $\N, \LB$ and $\al$, to achieve a vanishing error probability. 

To construct $\cM$, we consider parameters $(Z,m)$, where 
we assume $m_t=m$ tests are performed only for  $t\in [Z]$ (i.e., $m_t=0$  for $t>Z$). 
%
Therefore, the total number of tests performed is $M=mZ$. The parameters $(m,Z)$ will be determined later to guarantee the identification {of the} unreliable workers with high probability. 


\noindent \textbf{Probabilistic construction for $\cM$ and $\sM$:} \textit{We generate $\cM \in \mathbb{F}^{Zm \times \N}$ randomly, where each entry is drawn from a Bernoulli distribution with parameter $\cq := \frac{\tht}{\LB}$ (where the parameter $\tht$ will be determined later), independent of all other entries. For this scheme, the corresponding sampling matrix $\sM$ is generated from $\cM$ as follows,}
\begin{equation} \label{eq:sampling}
{\sM_{\mathcal{T}_z,w} = \begin{cases}
   \cM_{\mathcal{T}_z,w}  &  \text{if } w \in \left([\N]\setminus \Ls\right) \cup \Lst{z}, \\
   \mathbf{0} & \text{if } w \in \Ls \setminus \Lst{z},
    \end{cases}  
    }
\end{equation}
where, $\mathcal{T}_z=[(z-1)m+1:zm]$ is the set of rows which correspond to time-slot {$z \in Z$.}

\noindent \textbf{The $(d,\epsilon)$-threshold decoder: } \textit{For $z \in [Z]$, we let ${\ytz \in \{0,1\}^m}$ be the test result vector of time-slot $z$, which is given by $\ytz = \sM_{\mathcal{T}_z,:} \odot \bm{x}$,
where $\mathcal{T}_z=[(z-1)m+1:zm]$ is the set of tests in time-slot $z$. Moreover, we define the score of a worker $w \in [n]$ in time-slot $z$ by $I_{w,z}$, which is given by
\begin{align} \label{eq:scorecalc}
    I_{w,z} = \begin{cases}
        1 &\text{if } \cM_{\mathcal{T}_z,w} \preceq  \ytz \text{ and } \cM_{\mathcal{T}_z,w} \neq \mathbf{0} , \\
        \epsilon & \text{if }\cM_{\mathcal{T}_z,w} = \mathbf{0}, \\
        0 & \text{if }\cM_{\mathcal{T}_z,w} \centernot \preceq  \ytz.
    \end{cases}
\end{align}
Then, a worker $w$ will be marked as unreliable (i.e., {${\hat{\bm{x}}_w}=1$})
if and only if
$I_w := \sum\nolimits_{z=1}^{Z} I_{w,z} \geq d$.
}
{Note that an unreliable worker's score is expected to be higher than a reliable worker's score. We leverage this fact to obtain the following result that} guarantees the success of the proposed decoder (with high probability) for {a proper} choice of parameters. 
\begin{theorem} \label{thm1}
For any arbitrary, but a-priori fixed, {$L$-sparse} vector $\bm{x}$ and any error exponent $\beta>0$, there exists a choice of parameters $\left(q,  m, Z, d, \epsilon \right)$ such that  the $(d,\epsilon)$-threshold decoder can decode $\bm{x}$ with  
%
error probability ${P_e=\P[\hat{\bm{X}}\neq \bm{x} | \bm{X}= \bm{x}] \leq \N^{-\beta}}$ with at most ${450 (1+\beta)\frac{\LB\log(\N)}{\al}}$ tests. 
\end{theorem}
{\begin{rem}
    In this work, we focus on noise-level-independent Bernoulli designs of the contact matrix $\cM$ with parameter $\cq=\frac{\tht}{\LB}$ (where $q$ does not depend on the noise parameter $\al$). Advantages of such a design include: (1) increased robustness against erroneous estimates of $\al$, and (2) no requirement of redesigning $\cM$ every time that $\al$ fluctuates. Under such a construction of $\cM$, it is not difficult to show an information theoretic converse of $O\left(\frac{\LB\log(n)}{\al \log(1/\al)}\right)$ tests for the group testing noise model considered in this work. The $(d,\epsilon)$-threshold decoder almost achieves this converse, except for a small factor of $\log(1/\al)$. Moreover, noise-level-dependent designs could lead to better achievable bounds, as it has been observed in other noise models~\cite{cheng2023,arpino2021} and we intend to explore such designs in the extended version of this paper. 
\end{rem}}
In the proof of Theorem~\ref{thm1}, we will show that under a proper choice of {$\left(q,  m, Z, d, \epsilon \right)$,} 
the error probability vanishes as $n^{-\beta}$. In particular, we set $(q,m,Z,\epsilon)=\left(\frac{\theta}{L}, \frac{L}{\theta}, {\frac{\lambda \log n}{\alpha}}, \theta \alpha\right)$, where $\lambda$ and  $\theta\in[0,1]$ are design parameters, to be determined later. Moreover, 
the value of $d$ is given in~\eqref{eq:d}.

We will need the following two propositions, which play an essential role in the proof. The proofs of the propositions are presented in Appendix~\ref{app:proof_prop1} {and Appendix~\ref{app:proof_prop2}.} 

\begin{prop} \label{prop:expected_scores}
    The total expected scores 
    for reliable and unreliable workers are given by,
    \begin{align}
        \mf&:=\mathbb{E}\!\left[I_{w}\! \bigm| \! w \notin \Ls\right] =  \!Z\left( \hh_L - (1\!-\!\epsilon)\left(1\!-\!\cq\right)^m\right),  \\
        \mm&:=\mathbb{E}\!\left[I_{w}\!\bigm|\! w \in \Ls\right] \!=\! Z\left(\al \!+\!(1\!-\!\al) \hh_{\LB-1}\!-\! (1\!-\!\epsilon)(1\!-\!\cq)^m\right),\nonumber
    \end{align}
    where for every integer {$x\in [0:L]$,} we have that
    \begin{align}\label{eq:def:h}
    \hh_x := \sum\nolimits_{\ell=0}^x {\binom{x}{\ell}} \al^\ell (1-\al)^{x-\ell} \left(1-\cq(1-\cq)^\ell\right)^m.
\end{align} 
\end{prop}
\begin{prop}\label{prop:bounds_hL}
Assume $qm \leq 1$. Then, $\hh_x$ {in~\eqref{eq:def:h}}  satisfies 
\begin{align}\label{eq:h:bound}
   (1\!-\!\cq)^m \leq \hh_{\LB}\! \leq (1\!-\!\cq)^m+ mLq^2 \al,
\end{align}\vspace{-7mm}
\begin{align}\label{eq:h:diff}
     {\text{and, }} \quad \hh_{\LB} - (1-\alpha)\hh_{\LB-1}\leq \alpha \left(1-mq(1-qL)/{2}\right).
\end{align}
\end{prop}


\begin{proof}[Proof of Theorem~\ref{thm1}]

There are two types of error associated with the $(d,\epsilon)$-threshold decoder: (i) a \emph{false alarm}, where a worker $w\in [n]\setminus \Ls$ is identified as unreliable, and (ii) a \emph{mis-detection} error, where an unreliable worker $w\in \cL$ is identified as reliable. Correspondingly, we define ${\Pp(w):= \P\left(I_w \geq d| w \notin \Ls\right)}$ and ${\Pn(w) = \P\left(I_w < d| w \in \Ls\right)}$, where $d$ is the decoder parameter which is set to 
\begin{align}\label{eq:d}
    d:=(1+\eta) \mf = (1+\eta)Z (\hh_L- (1-\epsilon)(1-q)^m),
\end{align}
where $\eta>0$ will be determined later. 
Since there are $\LB$ unreliable and $\N-\LB$ reliable workers, using the union bound, 
the total error probability is {upper} bounded by 
\begin{align} \label{eq:errors}
\Pe \leq \sum_{w\in [n]\setminus \Ls} \Pp(w) + \sum_{w\in \Ls} \Pn(w).
\end{align}
In the following, we will bound both probabilities of error for the regime of parameters of interest. We start with 
\begin{align}\label{eq:Pp}
    &\Pp(w) \!=\! \mathbb{P}\left(I_w \! \geq\! d \bigm| w \notin \Ls \right)
    \stackrel{ }{=}\mathbb{P}\left(I_w \!\geq\! (1+\eta) \mf\bigm| w \notin \Ls \right) \nonumber \\
    &\hspace{-1pt}\stackrel{\rm{(a)}}{\leq} \hspace{-1pt}\exp\left(\hspace{-1pt}-\frac{\eta^2 }{2\hspace{-1pt}+\hspace{-1pt}\eta}\mf\hspace{-1pt}\right)  
    \hspace{-2pt}= \hspace{-2pt}\exp\left(-\frac{\eta^2}{2\hspace{-1pt}+\hspace{-1pt}\eta} Z\left( \hh_L \hspace{-1pt}-\hspace{-1pt} (1\hspace{-1pt}-\hspace{-1pt}\epsilon)\left(1\hspace{-1pt}-\hspace{-1pt}\cq\right)^m\right)\hspace{-1pt}\right) \nonumber \\
     &
     \stackrel{\rm{(b)}}{\leq} \exp\left(-\frac{\eta^2}{2+\eta} Z \epsilon\left(1-\cq\right)^m\right) \nonumber\\
     &\stackrel{\rm{(c)}}{\leq} \exp\left(\!-\frac{\eta^2}{2\!+\!\eta} \theta \exp\left(\!-\frac{qm}{1\!-\!q}\right)Z\alpha\right)
 \!\stackrel{\rm{(d)}}{=} \!\exp\left(-\zeta_1 Z \alpha \right), 
\end{align}
where: $\rm{(a)}$ is due to the Chernoff bound; in~$\rm{(b)}$ we used Proposition~\ref{prop:bounds_hL}; $\rm{(c)}$ follows from the facts that $\epsilon= \theta \alpha$, and $1-x \geq \exp\left(-\frac{x}{1-x}\right)$ which holds for $x<1$; and  in $\rm{(d)}$, we used the choice of parameters $(q,m)=\left(\frac{\theta}{L}, \frac{L}{  \theta}\right)$ 
which yields $1-q\geq 1-\theta$ and
\[
\zeta_1 = \frac{\eta^2}{2+\eta} \theta \exp\left(-\frac{1}{1-q}\right)  \geq  \frac{\eta^2}{2+\eta} \theta \exp\left(-\frac{1}{1-\theta}\right). 
\]
Similarly, for $\Pn(w)$ we can write
\begin{align}\label{eq:P_:1}
    \Pn(w) &= \P(I_w < d | w\in \cL) = \P(I_w < (1-\delta) \mm | w\in \cL) \nonumber\\
    &\leq \exp\left(-\frac{1}{2}\delta^2 \mm\right),
\end{align}
where $\delta := \frac{\mm - d}{\mm}$. 
We note that
\begin{align}\label{eq:P_:2}
    &\mm - d = \mm - (1+\eta)\mf\nonumber\\
    &= Z \left[\alpha -  \eta(\hh_L - (1-\epsilon) (1-q)^m) -( \hh_\LB- (1-\alpha) \hh_{\LB-1})\right]  \nonumber\\
    &\stackrel{\rm{(e)}}{\geq} Z \left[\alpha -\eta (\epsilon (1-q)^m + mLq^2 \alpha) - \alpha \left(1-\frac{mq(1-qL)}{2}\right) \right]\nonumber\\
    &\stackrel{\rm{(f)}}{=} Z\alpha [-\eta \theta (1-q)^m - \eta mLq^2 + mq/2 -mLq^2/2] \nonumber\\
    &= \zeta_2 Z\alpha ,
\end{align}
where $\rm{(e)}$ follows from  Proposition~\ref{prop:bounds_hL}, and in~$\rm{(f)}$ we plugged in $\epsilon=\theta \alpha$. Note that 
\begin{align*}
    \zeta_2 \hspace{-0pt}&=\hspace{-0pt} \frac{mq}{2} \hspace{-0pt}-\hspace{-1pt}\eta \theta (1\hspace{-0pt}-\hspace{-0pt}q)^m \hspace{-0pt}-\hspace{-0pt} \Big(\hspace{-0pt}\eta\hspace{-0pt}+\hspace{-0pt}\frac{1}{2}\hspace{-0pt}\Big)mLq^2 
    \nonumber\\ &   \hspace{-0pt}\stackrel{\rm{(g)}}{\geq}\hspace{-0pt} \frac{1}{2} \hspace{-0pt}-\hspace{-0pt} \theta\Big(\hspace{-0pt} {\eta \exp(-1)} \hspace{-0pt}+\hspace{-0pt} \Big(\hspace{-0pt}\eta\hspace{-0pt}+\hspace{-0pt}\frac{1}{2}\hspace{-0pt}\Big)\hspace{-0pt}\Big)\hspace{-0pt}, 
\end{align*}
where in~$\rm{(g)}$ we used the facts that $(q,m)=\left(\frac{\theta}{L}, \frac{L}{\theta}\right)$ and ${(1-q)^m \leq \exp(-mq)=\exp(-1)}$. 


Moreover, we have that
\begin{align}\label{eq:P_:3}
    \mm &= Z[\alpha + (1-\alpha) \hh_{\LB-1} - (1-\epsilon) (1-q)^m] \nonumber\\
    & \stackrel{\rm{(h)}}{\leq} Z[\alpha + \epsilon (1-q)^m + mLq^2 \alpha ] =  \zeta_3 Z\alpha,
\end{align}
where the inequality in~$\rm{(h)}$ follows from the chain of inequalities ${(1-\alpha)\hh_{\LB-1} \leq \hh_{\LB-1} \leq \hh_{\LB} \leq (1-q)^m+mLq^2 \alpha}$, which is implied by Proposition~\ref{prop:bounds_hL}. 
Note that since $\epsilon=\theta \alpha$, we have 
\begin{align*}
    \zeta_3 = 1+\theta (1-q)^m +mLq^2 \stackrel{\rm{(i)}}{\leq} 1+ \theta\exp(-1)+\theta,
\end{align*}
where $\rm{(i)}$ follows since $(q,m)=\left(\frac{\theta}{L}, \frac{L}{  \theta}\right)$ and from the fact that ${(1-\cq)^m \leq \exp(-m\cq)=\exp(-1)}$. Therefore, plugging~\eqref{eq:P_:2} and~\eqref{eq:P_:3} into~\eqref{eq:P_:1}, we get
\begin{align}\label{eq:Pn}
    \Pn(w) &\leq \exp\left(-\frac{1}{2} \frac{(\mm - d)^2}{\mm}\right) \nonumber\\
    &\leq \exp \left(-\frac{1}{2} \frac{\zeta_2^2 Z^2\alpha^2}{\zeta_3 Z \alpha}\right) = \exp(-\zeta_4 Z \alpha),
\end{align}
where $\zeta_4 = \zeta_2^2/2\zeta_3$. Now,  setting $(\theta,\eta)=(0.15, 1)$, we have $\zeta_1,\zeta_4> \zeta:=0.015$. 
Plugging~\eqref{eq:Pp} and~\eqref{eq:Pn} in~\eqref{eq:errors}, we get
\begin{align*}
    \Pe &\leq (n-\LB)\exp\left( -\zeta_1 Z\alpha\right) + \LB \exp\left(-\zeta_4 Z\alpha\right)\nonumber\\
    & \stackrel{\rm{(j)}}{{\leq}} n \exp\left( -\zeta \frac{\lambda \log n}{ \alpha}\alpha\right)= \exp\left( - \left(\zeta\lambda-1\right) \log n \right)
    \stackrel{\rm{(k)}}{=} n^{-\beta},
\end{align*}
where $\rm{(j)}$ holds 
for $\zeta_1, \zeta_4 >\zeta$ and 
 $Z=\frac{\lambda\log n}{\alpha}$, and~$\rm{(k)}$   follows from the fact that $\lambda=(1+\beta)/\zeta$. 
Note that the total number of tests is 
   $ mZ = \frac{L}{\theta} \frac{\lambda \log n}{\alpha} < 450 (1+\beta) \frac{L \log n}{\alpha}$.
This completes the proof {of} Theorem~\ref{thm1}.
\end{proof}

\section{Distributed Scheme}\label{section:distributedScheme}
In this section, we present the distributed computing scheme by dividing the scheme into the following three subsections.
\subsection{Generator Matrix for Encoding}
To obtain the generator matrix for encoding, we start with the group testing contact matrix $\cM \in \{0,1\}^{M \times \N}$ that we obtained in  Section~\ref{section:group_testing}, where $M=mZ$. We design a random parity matrix $\pM\in \mathbb{F}^{M\times n}$, given by $\pM_{i,j} = \cM_{i,j} \cdot \sQ_{i,j}$, where the entries of $\sQ\in \mathbb{F}^{M\times n}$ are chosen uniformly and independently at random from $\mathbb{F}\setminus\{0\}$. 
We will use the linear code induced by $\pM$ for pre-coding of the matrix $\mathbf{B}$. To this end, let $\sG\in\mathbb{F}^{k\times n}$ be the generator matrix of the \emph{systematic} code induced by the parity-check matrix $\pM$, that is, 
\begin{align*}
    \sG=\begin{bmatrix} \mathbf{I}_{k\times k} \big| \mathbf{R}_{k\times (n-k)} \end{bmatrix},
\end{align*}
for some matrix $\mathbf{R} \in \mathbb{F}^{k \times (n-k)}$ where\footnote{We note that in general it is rather unlikely, but possible, that $\pM$ is not full-rank.  In spite of that, we can always find  $k= n-M$ linearly independent vectors in $\mathbb{F}^n$ which are orthogonal to the rows of $\cM$.} 
$k=n-M$ and $\sG$ satisfies $\pM \cdot \sG^T =\mathbf{0}$.

Next, we use the matrix $\mathsf{G}$ for encoding  the original matrix $\mat{t}\in\mathbb{F}^{r\times c}$. To this end, $\mat{t}$ is first divided horizontally into $k$ equal parts $\mat{t}_1, \mat{t}_2, \cdots, \mat{t}_k$. Then, for  $w \in [\N]$ we generate  the matrix $\mathbf{W}^{(w)} \in \mathbb{F}^{s\times c}$ (where $s=r/k$), given by
\begin{equation}
    \mathbf{W}^{(w)} = \sum\nolimits_{j=1}^{k} \mathsf{G}_{j,w} \mat{t}_j, 
\end{equation}
and send it to worker $w$, who is responsible for computing the product ${\mathbf{a}_t^{(w)} = \mathbf{W}^{(w)}\cdot \vect{t}}$. 

\subsection{Group Testing for Identifying the Unreliable Workers}
Worker $w$ will return the result $\tilde{\mathbf{a}}_t^{(w)}$ {in~\eqref{eq:z-channel},} which may or may not be equal to the expected result $\mathbf{a}_t^{(w)}$,  depending on whether $w \in \Lst{t}$ or not (i.e., worker $w$ is attacked in time-slot $t$ or not;  see~\eqref{eq:z-channel}). The server can not determine if $\tilde{\mathbf{a}}_t^{(w)}$ is correct or incorrect only from the information provided by worker $w$. However, the set of correct results $\{\mathbf{a}_t^{(w)}:w \in [\N]\}$ satisfy $M$ parity equations, corresponding to the $M$ rows of the parity-check matrix $\pM$. To formalize this, we define a parity check function as follows.\\
\textbf{Parity Check Function}: Consider some $i \in [M]$ and a set of workers $\cU_i = \supp\left(\pM_{i,:}\right)$ with the corresponding results $\{\tilde{\mathbf{a}}_t^{(j)}\hspace{-1pt}: \hspace{-1pt}j \hspace{-1pt}\in\hspace{-1pt} \cU_i\}$. We define the parity check function $\Gamma_t\left(\cU_i\right)$ as
\begin{equation}\label{eq:parity-check-function}
   \Gamma_t(\cU_i) :=  \sum_{j\in \cU_i} \pM_{i,j} \tilde{\mathbf{a}}_t^{(j)},
\end{equation}
which can be used to check if there is any attacked (and hence unreliable) worker in the set $\cU_i$ that returned an incorrect result in time-slot $t\in [T]$. The following lemma, the proof of which is in Appendix~\ref{app:proof_lemma1}, states important properties of $\Gamma_t(\cdot)$. 
%

\begin{lemma} \label{Lemma:paritycheckfunction}
In the finite field $\mathbb{F}$, for every $i \in [M]$ and set $\cU_i = \supp(\pM_{i,:})$,  we have \begin{align}\label{eq:parity-cond}
\begin{split}
   &\P[ \Gamma_t\left(\cU_i\right)=\mathbf{0} | \cU_i \cap \Lst{t}  = \varnothing] = 1,\\
   &\P[ \Gamma_t\left(\cU_i\right)=\mathbf{0} | \cU_i \cap \Lst{t}  \neq \varnothing] = \frac{1}{|\mathbb{F}|}.
   \end{split}
\end{align}
\end{lemma}


We now discuss the use of the parity check function to identify the set of unreliable workers.

\noindent\textbf{Probabilistic Group Testing:} In time-slot $z \in [Z]$, the server computes $\hat{\bm{y}}_{j}$ for $j \in \mathcal{T}_z={[(m-1)z+1:mz]}$, (where the parameters $Z$ and $m$ are specified in Section~\ref{section:group_testing}) as follows,
\begin{align}\label{eq:result}
    \hat{\bm{y}}_j = \begin{cases}
        0 & \text{ if }\Gamma_z\left(\cU_j\right) = \mathbf{0}, \\
        1 & \text{ otherwise},
    \end{cases}
\end{align}
\noindent where $\cU_j = \supp\left(\pM_{j,:}\right) = \supp\left(\cM_{j,:}\right)$, since $\cM_{j,w}=0$ if and only if $\pM_{j,w}=0$. Hence, the parity equation on $\cU_j$ is equivalent to a group test on $\cU_j$. However, due to~\eqref{eq:parity-cond}, we get
\[\mathbb{P}\left({\hat{\bm{y}} = \bm{y}}\right) = \mathbb{P}\left({\hat{\bm{y}} = \sM \odot \bm{x}}\right) \geq \left(1-{1}/{|\mathbb{F}|}\right)^M,\]
or equivalently, $\mathbb{P}\left({\hat{\bm{y}} \neq \bm{y}}\right) \leq 1-\left(1-{1}/{|\mathbb{F}|}\right)^M\leq \frac{M}{|\mathbb{F}|}$. Therefore, the server can use the \mbox{$(d,\epsilon)$-threshold} decoder to identify all the unreliable workers $\Ls$ with an error probability less than $\N^{-\beta} + \frac{M}{|\mathbb{F}|}$, where the $\N^{-\beta}$ and $\frac{M}{|\mathbb{F}|}$ represent the upper bounds on the errors due to probabilistic group testing and the probabilistic nature of the parity-check function, respectively.

\begin{rem}
    Here, we need to choose  a finite field with $|\mathbb{F}| \gg M$ to guarantee the success of the algorithm with high probability. However, using a more sophisticated analysis for group testing, we can consider the underlying model similar to~\cite{Atia2009}, in which the results of the tests are passed through a $Z$-channel, i.e., $\P(\hat{\bm{y}}_j=0 | \bm{y}_j=1)= 1/|\mathbb{F}|$. 
\end{rem}

\noindent \textbf{Average Computational Cost of Identifying Unreliable Workers: } On average, each row of $\pM$ contains $\cq \N = \frac{\tht \N}{\LB}$ non-zero elements and $\tilde{\mathbf{a}}_t^{(w)}$ has $s = \frac{r}{k}$ elements. Therefore, the computational cost of $\Gamma_t(\cU_j)$ for $j \in [M]$ is ${O\left(\frac{\tht \N r}{k\LB}\right)=O\left(\frac{\N r}{k\LB}\right)}$. Moreover, a total of $M = O\left(\frac{\LB\log(\N)}{\al}\right)$ such tests are performed. Therefore, the total cost of identifying the unreliable workers is $O\left(\frac{r \N \log(\N) }{k\al}\right)$. Note that this does not scale with $T$, the number of vectors to be multiplied by $\mathbf{B}$. Moreover, for the interesting regime where $\frac{L\log n}{\al} = o(n)$, we have $M=o(n)$ and hence $k=n-M=\Theta(n)$. Consequently, in this regime, the total cost of identifying the unreliable workers will be $O\left(\frac{ r \log(\N) }{\al}\right)$.
\subsection{Reconstruction of Incorrect Results and Decoding}
Note that after performing group testing, the server knows the set of unreliable workers $\Ls$ with a probability of error less than $\N^{-\beta} + \frac{M}{|\mathbb{F}|}$. Since the generator matrix is of the form ${\mathsf{G} = \begin{bmatrix} \mathbf{I}_{k\times k} \big| \mathbf{R}_{k\times (n-k)} \end{bmatrix}}$, the first $k$ workers received matrices $\{\mathbf{B}_1, \mathbf{B}_2, \cdots, \mathbf{B}_k\}$ and were expected to compute ${\{{\mathbf{a}}_t^{(1)} = \mathbf{B}_1 \cdot \vect{t}, {\mathbf{a}}_t^{(2)} = \mathbf{B}_2 \cdot \vect{t}, \cdots, {\mathbf{a}}_t^{(k)} = \mathbf{B}_k \cdot \vect{t}\}}$ in time-slot $t$. Now, if the server was lucky and all of the first $k$ workers were reliable, that is,  $\Ls \cap [k] = \varnothing$, then $\tilde{\mathbf{a}}_t^{(w)}= {\mathbf{a}}_t^{(w)}$ for all $w \in [k]$ and therefore the correct matrix-vector product $\mat{t}\cdot \vect{t}$ is directly given by the result of the first $k$ workers, without any extra decoding computation cost at the server. However, if one or more of these $k$ workers were unreliable, that is, $\Ls \cap [k] \neq \varnothing$, then $\tilde{\mathbf{a}}_t^{(w)}$ may not be equal to ${\mathbf{a}}_t^{(w)}$ for $w \in \Ls \cap [k]$. In the following, we show that the server can reconstruct the correct answers ${\mathbf{a}}_t^{(w)}$ for all the unreliable workers $w \in \Ls$ with vanishing probability of error.
    Towards this end, we first define a reconstruction criterion as follows.

\noindent\textbf{Reconstruction Criterion}: \textit{For any worker $w\in \Ls$, the parity matrix $\pM$ satisfies the reconstruction criterion if $\pM$ has a row $i$ whose support ${\cU_i = \supp\left(\pM_{i,:}\right)}$ has the following properties: (i) $\cU_i \cap \Ls = w$ and (ii) ${\cU_i\cap \left([\N]\setminus \Ls \right)\neq \varnothing}$.}

The following lemma shows that if $\pM$ satisfies the reconstruction criterion for a worker $w\in \Ls$, then the server can reconstruct the correct result ${\mathbf{a}}_t^{(w)}$ of $w$. The proof of the lemma is presented in  Appendix~\ref{app:proof_lemma2}.
 \begin{lemma} \label{corollary:reconstruction}
 If the parity matrix $\pM$ satisfies the reconstruction criterion for a worker $w\in \Ls$,  then the correct result ${\mathbf{a}}_t^{(w)}$ of $w$ can be reconstructed as follows,
        \begin{align*}
              {\mathbf{a}}_t^{(w)} = \left(\pM_{i,w}\right)^{-1}\left(\pM_{i,w} \tilde{\mathbf{a}}_t^{(w)}-\Gamma_t\left(\cU_i\right)\right),
        \end{align*}
        where $i\in[M]$ is a row for which the two conditions of the reconstruction criterion are satisfied and  $\cU_i = \supp(\pM_{i,:})$.
    \end{lemma}

In the next lemma, we show that $\pM$ satisfies the reconstruction criterion for all unreliable workers $w \in \Ls$ with vanishing probability of error. The proof of the lemma is presented in Appendix~\ref{app:proof_lemma3}.
    \begin{lemma}\label{Lemma:reconstruction}
       The parity matrix $\pM$ constructed from the group testing contact matrix $\cM$ satisfies the reconstruction criterion for all the unreliable workers $w \in \Ls$, with an error probability less than $\N^{-\beta}$.
    \end{lemma}

Since the reconstruction criterion holds for every $w\in \Ls$ with a probability of error less than $\N^{-\beta}$, the server can reconstruct the results of all the unreliable workers with an error probability less than $\N^{-\beta} + \frac{M}{|\mathbb{F}|} +\N^{-\beta}= 2 \N^{-\beta}+\frac{M}{|\mathbb{F}|}$, where $\N^{-\beta} + \frac{M}{|\mathbb{F}|}$ and $\N^{-\beta}$ are upper bounds on the error probabilities for the identification of the unreliable workers and for the reconstruction of the correct results of the unreliable workers, respectively. Once the correct results are reconstructed, the server can use the first $k$ workers' results to obtain the matrix-vector product $\mat{t}\cdot \vect{t}$.

\noindent \textbf{Average Computational Cost of Decoding: } On average, each row of $\pM$ contains $\cq \N = \frac{\tht \N}{\LB}$ non-zero elements and $\tilde{\mathbf{a}}_t^{(j)}$ has $s = \frac{r}{k}$ elements. Therefore, the computational cost of reconstructing each unreliable worker's answer by Lemma~\ref{corollary:reconstruction} is $O\left(\frac{\tht \N r}{k\LB}\right)=O\left(\frac{\N r}{k\LB}\right)$. Moreover, the total cost of reconstructing all of the $\LB$ unreliable workers' results is $O\left(\frac{\N r}{k}\right)$. Once the unreliable workers' results are reconstructed, the matrix-vector product
{$\mat{t}\cdot\vect{t}$} is given by stacking the correct results of first $k$ workers, which has no additional cost. Therefore, the computational cost of decoding is $O\left(\frac{\N r}{k}\right)$.

\newpage 
\begingroup
\let\cleardoublepage\clearpage

\endgroup

\newpage

\newpage

\clearpage
\appendices

\section{Proof of Proposition~\ref{prop:expected_scores}} \label{app:proof_prop1}


We start by noting that
\begin{align}\label{eq:Ieps}
    \mathbb{P}\left(I_{w,z}\hspace{0pt} = \hspace{-0pt}\epsilon\right) 
     &=\mathbb{P}\left(\cM_{\mathcal{T}_z,w} = \mathbf{0}\right)\nonumber\\
     &= \prod_{j\in \mathcal{T}_z}\hspace{-1pt} \mathbb{P}\left(\cM_{j,w} = 0\right) = (1-\cq)^m. 
\end{align}
For a reliable worker $w \notin \Ls$, {we have that}
\begin{align}\label{eq:I1-notL:1} 
    \P&\left(I_{w,z} = 1\!\bigm|\! w \notin \Ls\right)  =\mathbb{P} \left(\cM_{\mathcal{T}_z,w}\!\preceq \ytz,  \cM_{\mathcal{T}_z,w} \neq \mathbf{0}
    \!\bigm|\! w \notin \Ls \right) \nonumber\\
    &  \stackrel{\rm{(a)}}{=} \P \left(\cM_{\mathcal{T}_z,w}\preceq \ytz \!\bigm|\! w \notin \Ls \right) -
    \P \left(\cM_{\mathcal{T}_z,w} = \mathbf{0}\!\bigm|\! w \notin \Ls \right)\nonumber\\
    & \stackrel{\rm{(b)}}{=} \sum_{\ell=0}^L \P \left(\cM_{\mathcal{T}_z,w}\preceq \ytz \!\bigm|\! |\cL_z|\!=\!\ell, w \!\notin\! \Ls \right)\P(|\cL_z|\!=\!\ell) 
    \nonumber\\ & \qquad 
    -(1-q)^m,
    \end{align}
where $\rm{(a)}$ follows from the fact that the event $\cM_{\mathcal{T}_z,w} = \mathbf{0}$ is a subset of the event $\cM_{\mathcal{T}_z,w}\!\preceq \ytz$, and in $\rm{(b)}$ we used {the law of total probability} and the fact that the event $\cM_{\mathcal{T}_z,w} = \mathbf{0}$ is independent of $w$, and it occurs with probability $(1-q)^m$, as shown in~\eqref{eq:Ieps}. Moreover, we have that
\begin{align}\label{eq:I1-notL:2}
    &\P \left(\cM_{\mathcal{T}_z,w}\preceq \ytz \!\bigm|\! |\cL_z|=\ell, w \notin \Ls \right)  \nonumber\\ 
    & \stackrel{\rm{(c)}}{=} \prod_{j \in \mathcal{T}_z} \mathbb{P}\left(\cM_{j,w}\preceq \bm{y}_j \!\bigm|\!|\Ls_z|=\ell, w \notin \Ls\right)\nonumber\\
    &= \prod_{j \in \mathcal{T}_z} \left(1-\mathbb{P} \left(\cM_{j,w} \succ \bm{y}_j\!\bigm|\!|\Ls_z|=\ell, w \notin \Ls \right)\right)  \nonumber\\
    &= \prod_{j \in \mathcal{T}_z} \left(1- \mathbb{P}\left(\cM_{j,w}=1, \bm{y}_j =0 \!\bigm|\! |\cL_z|=\ell, w\notin \Ls \right)\right) \nonumber \\
    &\stackrel{\rm{(d)}}{=} \!\!\prod_{j \in \mathcal{T}_z} \!\left(\hspace{-1pt}1\!-\! \mathbb{P}\left(\cM_{j,w}\!=\!1, \big\{\cM_{j,u} \!=\!0;  u\in \cL_z\big\} \hspace{-3pt}\bigm|\! |\cL_z|\!=\!\ell, w\hspace{-1pt}\notin\hspace{-1pt} \Ls \right)\hspace{-2pt}\right)\nonumber\\
    & = \prod_{j \in \mathcal{T}_z} \left(1-\cq(1-\cq)^\ell\right)  = \left(1-\cq(1-\cq)^\ell\right)^m,
\end{align}
where $\rm{(c)}$ follows from the fact that the entries of $\cM$ are generated independently and therefore, all of its rows are independent of each other, and in~$\rm{(d)}$ we used the fact that in the presence of $|\cL_z|=\ell$ 
{unreliable attacked workers,}
$\by_j=0$ if and only if 
{none of the unreliable attacked workers}
are selected in group $j$. 
Note that each unreliable worker {will be attacked} in each time-slot with probability $\alpha$, independent of all the other workers. Hence, $\P(|\cL_z| = \ell) = \binom{L}{\ell} \alpha^\ell {(1-\alpha)}^{L-\ell}$. Plugging this and~\eqref{eq:I1-notL:2} into~\eqref{eq:I1-notL:1}, we arrive at
\begin{align}\label{eq:I1-notL:3}
    \P&\left(I_{w,z} = 1\!\bigm|\! w \notin \Ls\right) \nonumber\\
    &= \sum_{\ell=0}^\LB \binom{L}{\ell} \alpha^\ell {(1-\alpha)}^{L-\ell} (1-q(1-q)^\ell)^m - (1-q)^m\nonumber\\
    &= \hh_L - (1-q)^m.
    \end{align}
    Therefore, from~\eqref{eq:Ieps} and~\eqref{eq:I1-notL:3}, we get 
    \begin{align}
        \E[I_w | w\notin \Ls] &= \sum_{z=1}^Z \E[I_{w,z}| w\notin \Ls] \nonumber\\
        &= Z\left(\hh_L - (1-q)^m + \epsilon (1-q)^m\right),
    \end{align}
which proves the first claim in Proposition~\ref{prop:expected_scores}.

Next, consider an unreliable worker $w \in \Ls$. First, using the law of total probability, we can write
\begin{align}\label{eq:I1-L:1}
    \mathbb{P}&\left(I_{w,z} = 1\bigm|  w \in \Ls \right) \nonumber\\
    =& \ \mathbb{P}\left(I_{w,z} = 1\bigm| w \in \Lst{z}, w \in \Ls \right) \P(w\in \cL_z| w\in \Ls)
     \nonumber\\
     & + \mathbb{P}\left(I_{w,z} = 1\bigm| w \notin \Lst{z}, w \in \Ls \right) \P(w\notin \cL_z| w\in \Ls) \nonumber\\
     \stackrel{\rm{(a)}}{=}&\  \alpha \mathbb{P}\left(I_{w,z} \hspace{-1pt}=\hspace{-1pt} 1 \hspace{-1pt}\bigm|\hspace{-1pt} w \in \Lst{z}\right) \hspace{-1pt} \nonumber
     \\&+\hspace{-1pt} (1\hspace{-1pt}-\hspace{-1pt}\alpha) \mathbb{P}\left(I_{w,z} \hspace{-1pt}=\hspace{-1pt} 1 \hspace{-1pt}\bigm|\hspace{-1pt} w \in \cL\setminus \Lst{z}\right),
\end{align}
where in~$\rm{(a)}$ we have used the fact that $\cL_z\subseteq \cL$. Then, similar to~\eqref{eq:I1-notL:1}, we can evaluate the first term in~\eqref{eq:I1-L:1} as
\begin{align}\label{eq:I1-L:2} 
&\mathbb{P} \left(I_{w,z} = 1 \bigm|  w \in \Lst{z} \right) \nonumber\\
&= \mathbb{P} \left(\cM_{\mathcal{T}_z,w} \preceq \ytz \!\bigm|\! w \in \Lst{z} \right) - \mathbb{P}\left( \cM_{\mathcal{T}_z,w}=\mathbf{0}\!\bigm|\! w \in \Lst{z} \right)\nonumber\\
&\stackrel{\rm{(b)}}{=} 1 - \left(1-\cq\right)^m,
\end{align}
where $\rm{(b)}$ holds 
{because when $w\in \cL_z$ (i.e., $w$ is an attacked unreliable worker), $\cM_{j,w}=1$ implies  $\bm{y}_j=1$.}

The second term in~\eqref{eq:I1-L:1} can be evaluated similar to~\eqref{eq:I1-notL:1} and~\eqref{eq:I1-notL:2}. However, since $w\in \cL\setminus \cL_z$, the size of $\cL_z$ can be any integer between $0$ and $L-1$. This leads to 
\begin{align}\label{eq:I1-L:3}
    \mathbb{P}\left(I_{w,z} = 1 \bigm| w \in \Ls\setminus\Lst{z} \right) = \hh_{\LB-1} - (1-q)^m. 
\end{align}
Plugging \eqref{eq:I1-L:2} and~\eqref{eq:I1-L:3} in~\eqref{eq:I1-L:1} we get 
\begin{align} \label{eq:I1-L:4}
    \mathbb{P}\big(I_{w,z} = 1 &\bigm| w \in \Ls \big) \nonumber\\ 
    &= \alpha \left(1-(1-q)^m\right) + (1-\alpha)\left(\hh_{\LB-1} - (1-q)^m\right)\nonumber\\
    &= \alpha + (1-\alpha) \hh_{\LB-1} - (1-q)^m.
\end{align}
Therefore, { using~\eqref{eq:I1-L:4} and \eqref{eq:Ieps}, we get}
\begin{align}
    \E[I_w|w\in \cL] &= \sum_{z=1}^Z  \E[I_{w,z}|w\in \cL] \nonumber\\
    &= Z\left( \alpha + (1-\alpha) \hh_{\LB-1} - (1-\epsilon)(1-q)^m\right),\nonumber
\end{align}
which completes the proof of Proposition~\ref{prop:expected_scores}. 

\section{Proof of Proposition~\ref{prop:bounds_hL}} \label{app:proof_prop2}
The lower bound follows easily from the definition of $\hh_\LB$. Since $(1-q)^{\ell} \leq 1$, we have that
\begin{align}
     \hh_\LB &= \sum_{\ell=0}^{\LB} \binom{\LB}{\ell} \al^{\ell} (1-\al)^{\LB-\ell}\left(1-\cq(1-\cq)^\ell\right)^m \nonumber \\
    &\geq \sum_{\ell=0}^{\LB} \binom{\LB}{\ell} \al^{\ell} (1-\al)^{\LB-\ell}\left(1-\cq\right)^m = \left(1-\cq\right)^m. \nonumber
\end{align}
In order to prove the upper bound, we first note that ${g(\ell):=(1-q(1-q)^\ell)^m}$ is a concave function provided that  $mq \leq 1$. Note, in fact, that 
\[
\frac{{\rm{d}}^2 g(\ell)}{{\rm{d}}\ell^2} \! =\! \frac{-mq (\log(1\hspace{-1pt}-\hspace{-1pt}q))^2 (1\hspace{-1pt}-\hspace{-1pt}q)^\ell  (1\hspace{-1pt}-\hspace{-1pt}mq (1\hspace{-1pt}-\hspace{-1pt}q)^\ell) {g(\ell)}}{(1-q(1-q)^{\ell})^{2}}\!\leq\! 0,
\]
when $mq \leq 1$. Then, using Jensen's inequality, we get 
\begin{align}
   \hh_{\LB} &= \sum_{\ell=0}^{\LB} \binom{\LB}{\ell} \al^{\ell} (1-\al)^{\LB-\ell}\left(1-\cq(1-\cq)^\ell\right)^m  \nonumber \\
   &\stackrel{\rm{(a)}}{=} \E[g(\ell)] \stackrel{\rm{(b)}}{\leq} g(\E[\ell]) = {g(\al {\LB)}} = \left(1-\cq(1-\cq)^{\LB \al}\right)^m \nonumber \\ &\stackrel{\rm{(c)}}{\leq} \left(1-\cq+\LB \al \cq^2\right)^m  \stackrel{\rm{(d)}}{\leq} \left(1-\cq\right)^m+m\LB \al \cq^2, 
\end{align}
where: the expectation in $\rm{(a)}$ is with respect to a binomial distribution $\mathsf{Bin}(L,\alpha)$, in~$\rm{(b)}$ we used Jensen's inequality,  $\rm{(c)}$ holds because $(1-\cq)^{\LB \al} \geq 1-\cq \LB\al$ by Bernoulli's inequality, and $\rm{(d)}$ follows from the inequality ${(x+y)^m \leq x^m + myx^{m-1}}$ along with the fact that ${(1-\cq)^{m-1} \leq 1}$. This completes the proof of~\eqref{eq:h:bound}. 

Next, we have that
\begin{align}\label{eq:P_unrel}
 &\hh_{\LB} - (1-\al)\hh_{\LB-1}\nonumber\\
    &  =\sum_{\ell={0}}^{\LB} \al^\ell (1-\al)^{\LB-\ell} \left(1-\cq(1-\cq)^\ell\right)^m \left(\binom{\LB}{\ell}-\binom{\LB-1}{\ell}\right) \nonumber \\
    & \stackrel{\rm{(e)}}{\leq}  \sum_{\ell={0}}^{\LB} \al^\ell (1-\al)^{\LB-\ell} \left(1-\cq(1-\cq\LB)\right)^m \binom{\LB-1}{\ell-1}\nonumber \\
    &= \alpha (1-q(1-qL))^m \sum_{k=0}^{L-1} \binom{L-1}{k} {\alpha^{k}} (1-\alpha)^{L-1-k} \nonumber\\
    & =
    \alpha (1-q(1-qL))^m \stackrel{\rm{(f)}}{\leq}   {\al} \left( 1-  \frac{mq}{2}(1-qL) \right ),
\end{align}
where in~$\rm{(e)}$ we used the Pascal's identity as well as the fact that $(1-q)^\ell \geq (1-q)^L \geq 1-qL$ for every $\ell\leq L$, and $\rm{(f)}$ follows from $(1-x)^m \leq 1-mx/2$, which holds true provided that $mx\leq 1$. 
This proves~\eqref{eq:h:diff} and concludes the proof of Proposition~\ref{prop:bounds_hL}.

\section{Proof of Lemma~\ref{Lemma:paritycheckfunction}} \label{app:proof_lemma1}

First the correct answers $\{{\mathbf{a}}_t^{(j)}: j \in \cU_i\}$ satisfy a parity corresponding to the row $\pM_{i,:}$ as follows,
\begin{align} \label{eq:correct_answerparity}
    \sum_{j\in \cU_i} \!\pM_{i,j} {\mathbf{a}}_t^{(j)} &\!\stackrel{\rm{(a)}}{=}\!\!\! \sum_{w\in [\N]}\! \pM_{i,w} \mathbf{a}_t^{(w)} \!=\!\!  \sum_{w\in [\N]} \! \pM_{i,w} \left(\mathbf{W}^{(w)} \cdot \vect{t}\right) \nonumber\\
   &= \sum_{w=1}^{\N} \pM_{i,w} \left(\sum_{j=1}^{k} \mathsf{G}_{j,w} \mat{t}_{j}\cdot \vect{t}\right) \nonumber \\
   &= \sum_{j=1}^{k} \sum_{w=1}^{\N} \left(\pM_{i,w}  \mathsf{G}_{j,w} \right) \mat{t}_{j}\cdot \vect{t} \stackrel{\rm{(b)}}{=} \mathbf{0},
\end{align}
where $\rm{(a)}$ follows since $\cU_i = \supp\left(\pM_{i,:}\right)$ and $\rm{(b)}$ follows because $\pM \cdot \sG^T =\mathbf{0}$. Moreover,
by~\eqref{eq:z-channel}, we know that ${\tilde{\mathbf{a}}_t^{(j)} = {\mathbf{a}}_t^{(j)}+\mathbf{z}_t^{(j)}}$ for $j \in \Lst{t}\cap \cU_i$ and $\tilde{\mathbf{a}}_t^{(j)} = {\mathbf{a}}_t^{(j)}$ for $j \in \cU_i \setminus \Lst{t}$. Therefore, we have 
\begin{align*}
\Gamma_t(\cU_i)&=\!\! \sum_{j\in \cU_i} \!\pM_{i,j}  \tilde{\mathbf{a}}_t^{(j)} = \!\!\sum_{j\in \cU_i} \!\pM_{i,j} {\mathbf{a}}_t^{(j)} + \!\!\!\! \sum_{j\in \Lst{t} \cap \cU_i}\!\!\!\! \pM_{i,j} {\mathbf{z}}_t^{(j)} \\
& \stackrel{\rm{(c)}}{=} \!\!\! \sum_{j\in \Lst{t} \cap \cU_i} \!\!\! \pM_{i,j} {\mathbf{z}}_t^{(j)}, 
\end{align*}
where $\rm{(c)}$ follows from~\eqref{eq:correct_answerparity}. Then, we have that
\begin{align}
\P\left[\Gamma_t(\cU_i)=\mathbf{0}\right] &= \P\left[\sum_{j\in \Lst{t} \cap \cU_i} \!\!\!\! \pM_{i,j} {\mathbf{z}}_t^{(j)} =\mathbf{0}\right] \nonumber\\
&= 
    \begin{cases}
         1 &\text{if } \Lst{t}\cap \cU_i = \varnothing,\\
         \frac{1}{|\mathbb{F}|} &\text{if }  \Lst{t}\cap \cU_i \neq \varnothing,
    \end{cases}
\end{align}
where the last equation follows because the unreliable workers are assumed to be non-colluding and unaware of the entries of $\pM$. This implies that each entry of the vector $\sum_{j\in \Lst{t} \cap \cU_i} \pM_{i,j} \mathbf{z}_t^{(j)}$ is a linear combination of the corresponding entries chosen by the attacked workers in $\cL_t \cap \cU_i$, and since these workers are not communicating, they choose their inserted noises independent of each other. Hence, their linear combination admits a uniform distribution over $\mathbb{F}$, which will be $0$ with probability\footnote{It may seem that this probability should be $1/|\mathbb{F}|^s$, since there are $s$ entries in the vector of interest. However, the attacked worker may only insert noise in a certain position, and not all the entries may be corrupted by their noise.} $1/|\mathbb{F}|$. 
This concludes the proof of Lemma~\ref{Lemma:paritycheckfunction}.

\section{Proof of Lemma~\ref{corollary:reconstruction}}\label{app:proof_lemma2}

 If the parity matrix $\pM$ satisfies the reconstruction criterion for a worker $w\in \Ls$, then there exists a row $i\in [M]$ {in $\pM$} which satisfies ${\cU_i \cap \Ls = w}$ and ${\cU_i\cap \left([\N]\setminus \Ls \right)\neq \varnothing}$. Therefore, we have that
\begin{align*}
&\left(\pM_{i,w}\right)^{-1}\left(\pM_{i,w} \tilde{\mathbf{a}}_t^{(w)}-\Gamma_t\left(\cU_i\right)\right) \\
    &=\left(\pM_{i,w}\right)^{-1}\left(\pM_{i,w} \tilde{\mathbf{a}}_t^{(w)}-\sum_{j\in \cU_i} \pM_{i,j} \tilde{\mathbf{a}}_t^{(j)}\right) \\
    &\stackrel{\rm{(a)}}{=} \left(\pM_{i,w}\right)^{-1}\left(\pM_{i,w} \tilde{\mathbf{a}}_t^{(w)}\!-\!\!\!\sum_{j\in [\N]\setminus\Ls} \pM_{i,j} \tilde{\mathbf{a}}_t^{(j)}  \!-\!\sum_{j\in \Ls} \pM_{i,j} \tilde{\mathbf{a}}_t^{(j)}\right) \\
    &\stackrel{\rm{(b)}}{=} \left(\pM_{i,w}\right)^{-1}\left(\pM_{i,w} \tilde{\mathbf{a}}_t^{(w)}\!-\!\!\!\sum_{j\in [\N]\setminus\Ls} \pM_{i,j} \tilde{\mathbf{a}}_t^{(j)} - \pM_{i,w} \tilde{\mathbf{a}}_t^{(w)}\right) \\
   &\stackrel{\rm{(c)}}{=} \left(\pM_{i,w}\right)^{-1}\left(-\sum_{j\in [\N]\setminus\Ls} \pM_{i,j} {\mathbf{a}}_t^{(j)} \right) \\
    &= \left(\pM_{i,w}\right)^{-1}\left(-\!\!\!\!\sum_{j\in [\N]\setminus\Ls}\!\!\! \pM_{i,j} {\mathbf{a}}_t^{(j)}\!-\pM_{i,w}{\mathbf{a}}_t^{(w)} \!+  \pM_{i,w}{\mathbf{a}}_t^{(w)}\right) \\
    &= \left(\pM_{i,w}\right)^{-1}\left(-\sum_{j\in\cU_i} \pM_{i,j} {\mathbf{a}}_t^{(j)}\right)\!+ {\mathbf{a}}_t^{(w)} \stackrel{\rm{(d)}}{=} {\mathbf{a}}_t^{(w)},
\end{align*}
where: 
$\rm{(a)}$ follows since $\cU_i = \supp\left(\pM_{i,:}\right)$,
$\rm{(b)}$ follows because $\cU_i \cap \Ls = w$ by the reconstruction criterion, $\rm{(c)}$ follows because for reliable workers ${j \in [\N]\setminus \Ls}$, $\tilde{\mathbf{a}}_t^{(j)} = {\mathbf{a}}_t^{(j)}$, and $\rm{(d)}$ holds because $\sum_{j\in\cU_i} \pM_{i,j} {\mathbf{a}}_t^{(j)}=\mathbf{0}$ by~\eqref{eq:correct_answerparity}.
This concludes the proof of Lemma~\ref{corollary:reconstruction}.

 \section{Proof of Lemma~\ref{Lemma:reconstruction}} \label{app:proof_lemma3}

 Note that each entry of the group testing contact matrix ${\cM \in \{0,1\}^{M \times \N}}$ is identically and independently (i.i.d.) generated and is equal to $1$ and $0$ with probabilities $\cq=\frac{\tht}{\LB}$ and $1-\cq$, respectively (where $\tht \leq 1$). Therefore, every entry in the corresponding parity matrix $\pM$ is non-zero and zero with probabilities $\cq$ and $1-\cq$, respectively. Let $E_w$ be the event that $\pM$ satisfies the reconstruction criterion for $w \in \Ls$ and let $\cU_j = \supp\left(\pM_{j,:}\right)$ for $j\in [M]$. Then, we have that
 \begin{align*}
     &\mathbb{P}\left(E_w^c\right) = \prod_{j=1}^{M} \mathbb{P}\left(\left(\cU_j \cap \Ls \neq w\right) \cup \left(\cU_{j} \cap ([\N]\setminus \Ls) = \varnothing \right) \right)\\
     &= \prod_{j=1}^{M} \left(1-\mathbb{P}\left(\left(\cU_j \cap \Ls = w\right) \cap \left(\cU_{j} \cap ([\N]\setminus \Ls) \neq \varnothing \right) \right)\right)\\
     &= \prod_{j=1}^{M} \left(\!1\!-\!\mathbb{P}\left(\pM_{j,w}\!\neq 0\right) \mathbb{P}\left(\pM_{j,\Ls\setminus\{w\}}\!\!= \mathbf{0}\right) \mathbb{P}\left(\pM_{j,[\N]\setminus\Ls}\!\neq \mathbf{0}\right) \right)\\
     &\stackrel{\rm{(a)}}{=} \left(1-\cq(1-\cq)^{\LB-1} \left(1-(1-\cq)^{\N-\LB}\right)\right)^M \\
     &\leq \left(1-\cq(1-\cq)^{\LB} \left(1-(1-\cq)^{\N-\LB}\right)\right)^M\\
     & = \left(1-\cq(1-\cq)^{\LB} + \cq(1-\cq)^{\N}\right)^M \\
     & \stackrel{\rm{(b)}}{\leq} \left(1-\cq(1-\cq)^{\LB} + \cq(1-\cq)^{2 \LB}\right)^M \\
     & \stackrel{\rm{(c)}}{\leq} \left(1-\cq\exp\left(-\frac{\tht}{1-\frac{\tht}{\LB}}\right) + \cq\exp\left(-2\tht\right)\right)^M \\
     & \stackrel{\rm{(d)}}{\leq} \left(1-\cq \exp\left(-\frac{\tht}{1-{\tht}}\right) + \cq\exp\left(-2\tht\right)\right)^{M}\\
     & = \left(1-
     \frac{\tht r_\tht}{\LB}\right)^{M} \stackrel{\rm{(e)}}{\leq} \exp\left(-\frac{M\tht r_{\tht}}{\LB}\right),
 \end{align*}
 \noindent where, $r_\tht \triangleq \exp\left(-\frac{\tht}{1-\tht}\right) - \exp\left(-2\tht\right)$ is a constant function of $\tht$. Step $\rm{(a)}$ follows because $\mathbb{P}\left(\pM_{j,w}\neq 0\right) = \cq$, ${\mathbb{P}\left(\pM_{j,\Ls}= \mathbf{0}\right) = (1-\cq)^{\LB-1}}$ and ${\mathbb{P}\left(\pM_{j,[\N]\setminus\Ls}\!\neq \mathbf{0}\right) = 1 - \mathbb{P}\left(\pM_{j,[\N]\setminus\Ls}\!= \mathbf{0}\right)} = 1-(1-\cq)^{\N-\LB}$; in $\rm{(b)}$, we have assumed\footnote{This assumption is justified because $ \LB = o(\N)$ for group testing to be relevant, because if $\LB \approx \N$, then testing one worker at a time is optimal and there is no requirement for group testing~\cite{Flodin2021},~\cite{Chan2011}.} that $\N \geq 2 \LB$; $\rm{(c)}$ follows from $\exp(-\frac{x}{1-x}) \leq 1-x \leq \exp(-x)$ which holds for $0<x<1$ and setting $\cq = \frac{\tht}{\LB}$ from the choice of parameters in Section~\ref{section:group_testing}; in $\rm{(d)}$, we have used the fact that $\LB \geq 1$; and finally $\rm{(e)}$ follows because $1-x \leq \exp(-x)$.
Moreover, using the union bound, the probability that the reconstruction criterion does not hold for at least one unreliable worker is upper bounded by
\begin{align*}
    &\sum_{w\in \Ls} \!\mathbb{P}\left(E_w^c\right)\!\leq\! \LB \exp\left(\!\!-\frac{M\tht r_\tht}{\LB}\right) \!\leq\! \N \exp\left(\!-\frac{M\tht r_\tht}{\LB}\right) \\
     &\!\stackrel{\rm{(f)}}{=} \N \exp\left(\!-\frac{450 (1\!+\!\beta) \log(\N)\tht r_\tht}{\al }\right) \!\stackrel{\rm{(g)}}{\leq}\! \N^{-\beta},
\end{align*}
where {in $\rm{(f)}$ we have used the value of $M$ found in Section~\ref{section:group_testing} and} in $\rm{(g)}$ we have used the choice of parameter ${\tht = 0.15}$ from Section~\ref{section:group_testing}. This concludes the proof of Lemma~\ref{Lemma:reconstruction}.
\end{document}